\documentclass[10pt]{amsart}
\usepackage[utf8]{inputenc}
\usepackage[table]{xcolor}
\usepackage{amsmath,amssymb,amsthm,enumerate}
\usepackage{epsf,epsfig,amsfonts,graphicx,color}
\usepackage{mathtools,enumitem,tcolorbox}
\usepackage{url,a4wide,epsfig}
\usepackage{ulem}
 \numberwithin{equation}{section}
\usepackage[utf8]{inputenc}
\usepackage{tikz,tikz-cd} 
\usepackage[margin=0.75in,foot=.25in]{geometry}
\usepackage[hidelinks]{hyperref}
\hypersetup{colorlinks,citecolor=black,filecolor=black,linkcolor=black,urlcolor=black}
\usepackage[singlelinecheck=false,justification=justified]{caption}
\usepackage{cleveref,pdfpages,comment}

\newcommand*\circled[1]{\tikz[baseline=(char.base)]{\node[shape=circle,draw,inner sep=2pt] (char) {#1};}}
\newcommand\op[1]{\mathop{\rm #1}\nolimits}
\newcommand\p{\partial}
\newcommand\R{{\mathbb R}}

\newcommand{\weg}[1]{}

\makeatletter
\newcommand*\owedge{\mathpalette\@owedge\relax}
\newcommand*\@owedge[1]{%
  \mathbin{%
    \ooalign{%
      $#1\m@th\bigcirc$\cr
      \hidewidth$#1\m@th\wedge$\hidewidth\cr
    }%
  }%
}
\makeatother

\theoremstyle{theorem}
\newtheorem{theorem}{Theorem}
\theoremstyle{definition}
\newtheorem{definition}[theorem]{Definition}

\theoremstyle{proposition}
\newtheorem{proposition}[theorem]{Proposition}
\theoremstyle{corollary}
\newtheorem{corollary}[theorem]{Corollary}

\title{Killing Tensors in Koutras--McIntosh Spacetimes}
\author{Boris Kruglikov and Wijnand Steneker}
 \address{Department of Mathematics and Statistics,
UiT the Arctic University of Norway, 9037 Troms\o, Norway.\newline
Emails: {\tt boris.kruglikov@uit.no} \ \& {\tt wijnand.steneker@gmail.com}}

\begin{document}

 \begin{abstract}
The Koutras--McIntosh family of metrics include conformally flat pp-waves and the Wils metric. 
It appeared in a paper of 1996 by Koutras--McIntosh as an example of a pure radiation 
spacetime without scalar curvature invariants or infinitesimal symmetries.
Here we demonstrate that these metrics have no ``hidden symmetries'', by which we mean
Killing tensors of low degrees. For the particular case of 
Wils metrics we show the nonexistence of Killing tensors up to degree 6.

The technique we use is the geometric theory of overdetermined PDEs and
the Cartan prolongation-projection method. Application of those allows to prove
the nonexistence of polynomial in momenta integrals for the equation of geodesics
in a mathematical rigorous way. 
Using the same technique we can completely classify all lower degree
Killing tensors and, in particular, prove that for generic pp-waves 
all Killing tensors of degree 3 and 4 are reducible. 
 \end{abstract}

\maketitle

\vspace{-0.5cm}

\section{Introduction}\label{S1}

\subsection{Formulation and motivation}\label{S11}

Polynomial integrals of Hamiltonian ODEs were actively studied in the XIX$^\text{th}$
century classical mechanics; 
in particular the existence of quadratic integral for the metric of the ellipsoid 
allowed Jacobi in 1836 to find an explicit formula for geodesics in terms of elliptic
functions.

This problem also appeared in general relativity: the famous metrics of Schwarzschild, Gödel and Kerr admit
polynomial integrals allowing to describe geodesics of the corresponding spacetimes in detail. 
Often integrals are conserved quantities related to Killing vectors via Noether's 
theorem, but sometimes there are higher degree integrals, known as Killing tensors. 
One of those is the Carter constant \cite{C,WP} reducing the geodesic motion to quadratures.

There exist obstructions to the existence of polynomial integrals:
according to \cite{KM2} a generic metric $g$ admits no such integrals even locally. 
It is thus important to realize the existence/nonexistence of Killing tensors for concrete metrics
from applications, see \cite{HM,KT,Ki,KM1,V}.

The following is the Koutras--McIntosh family of spacetimes for $(a,b)\neq(0,0)$:
 \begin{equation}\label{KMm}
g = 2(ax+b)\,du\,dw -2aw\,dx\,du +\bigl(f(u)(ax+b)(x^2+y^2)-a^2w^2\bigr)\,du^2 -dx^2 -dy^2.
 \end{equation}
These metrics were shown in \cite{KoMc} to possess neither invariants nor symmetries. 
The first property means that all polynomial curvature invariants, i.e., complete contractions of tensor products 
of the Riemann tensor and its covariant derivatives $\nabla_{i_1}\!\cdots\nabla_{i_s}R_{abcd}$, 
vanish and so cannot be used to distinguished $g$ from the Minkowski metric. 

These are so-called VSI (vanishing scalar invariants) spaces that received considerable attention in recent time \cite{PPCM}. 
They belong to a more general class of spacetimes not separated by their scalar curvature invariants \cite{CHP,CHPP}, 
which in dimension 4 were proven to be of degenerate Kundt type. Note that Kundt spaces can be distinguished by their
Cartan \cite{MMC} or differential \cite{KS} invariants, see \cite{KMS} for a comparisson. 

The second property above means there are no Killing vectors, or linear integrals, for 
\eqref{KMm}. 
In this paper we show that it also does not possess ``hidden symmetries'', by 
which we mean Killing tensors of low degrees. 

Note that the nonexistence of Killing tensors is important in several applications. For instance, it is
necessary for linearization stability of Einstein’s equations \cite{AM} and also for the inverse problem
in tensor tomography \cite{PSU}. Thus, even though Killing tensors do not have direct geometric interpretation
(as noticed by Penrose and Walker \cite{WP}, see however \cite{CGH}) 
their existence or nonexistence carries certain dynamical implications.

\subsection{Main results}\label{S12}

Metric \eqref{KMm} is conformally flat (but nonflat for $f\neq0$) 
and describes pure radiation, satisfying Einstein’s field equations of the type
$R_{ab}=\phi\, l_al_b$ for a null vector field $l$ and a scalar field $\phi$. 

Metric \eqref{KMm} for $a=0,b=1$ is a pp-wave, possessing 6 Killing vectors and 1 homothety 
except for special cases $f(u)=c$ and $f(u)=c/u^2$,
where the number of Killing vectors increases to 7 \cite{SG,KT} and the homothety persists. 
We will examine the existence of higher order
Killing tensors (up to degree 4) and for specific cases $f(u)=cu^m$, $m=0,1,2,-2$ 
we prove that there is only one such irreducible quadratic tensor.

Metric \eqref{KMm} for $a=1,b=0$ defines the Wils spacetime \cite{W}. This metric is known to have 
no Killing vectors or homotheties for general functional parameter $f(u)$, 
so we examine it for the existence of higher order Killing tensors. 
It turns out that up to order 6 no irreducible Killing tensors exist (that is with the exception of powers of
the Hamiltonian and combinations with Killing vectors when they exist).
These results are presented in Section \ref{S3}. 

For the general Koutras--McIntosh family we deduce the following statement:

 \begin{theorem} \label{Th1}
For generic numerical parameters $a,b$ and functional parameter $f(u)$ the spacetime \eqref{KMm}
possesses no Killing tensors up to degree 6 except for energy and its powers $H$, $H^2$ and $H^3$.
 \end{theorem}

Here and below genericity of $f(u)$  is understood in $C^{k+1}$ topology, where $k$ is the prolongation
level determined by Algorithm 1 of \S\ref{S24}, where the matrix $M_k$ depends on the jet $j^{k+1}f$.
Table 2 shows values of $k$ for degrees $d\leq6$.

We can be more specific on the exceptional values of the involved parameters. To find those that allow Killing vectors 
one may follow the general approach with metric invariants via the Cartan-Karlhede algorithm 
\cite{Kar}, however our 
method with counting compatibility conditions via the coefficient matrix of the prolonged PDE 
system gives an alternative and implies the following results.

 \begin{theorem} \label{Th2} 
Metrics \eqref{KMm} possess Killing vectors if and only if either $a=0$ (then rescale $b\to1$), 
so that the spacetime is a plane wave, or $b=0$ (then rescale $a\to1$), so that $g$ is Wils metric 
with $f(u)=(c_0+c_1u+c_2u^2)^{-2}$.
 \end{theorem}

The same approach but with much heavier computations yields the following results.

 \begin{theorem} \label{Th3} 
Metrics \eqref{KMm} possess Killing 2-tensors different from the Hamiltonian $H$ in the same range of parameters
as for the Killing vectors, i.e., either $a=0$ or $b=0$, $f(u)=(c_0+c_1u+c_2u^2)^{-2}$.
 \end{theorem}

The proofs and further specifications will be given in Section \ref{S3}. 
The Maple \& LinBox worksheets, which demonstrate our computations, 
can be found in a supplement to the arXiv version of this paper.

\section{Geometric Theory of PDEs}\label{S2}

We start with the general setup.
Let $(N^n,g)$ be a pseudo-Riemannian manifold. In this section we formalize searching for Killing tensors 
 (or polynomial integrals of the geodesic flow on the tangent bundle but we work on the cotangent
 bundle using raising/lowering indices with the metric $g$)
via a compatibility analysis of an overdetermined PDE system and discuss the prolongation-projection technique.

\subsection{Hamiltonian formalism}\label{S21}

The energy function $H=\frac12\|p\|^2_g$ writes in local coordinates
 \begin{equation*}
H(x,p) = \frac12 g^{ij}(x) p_i p_j \hspace{1cm} [g^{ij}]=[g_{ij}]^{-1}.
 \end{equation*}
It is well-known that geodesics of $g$ are projections to the base $N$ of trajectories 
of the corresponding Hamiltonian vector field $X_H=\omega^{-1}dH$ on $T^*N\stackrel{g}\simeq TN$,
where $\omega$ is the canonical symplectic form on the cotangent bundle.

Integrating the equations of geodesics requires conserved quantities for this Hamiltonian system.
A function $I:T^*N\to\R$ is an \textit{integral} (of motion) $X_H(I)=0$ if it Poisson commutes 
with the Hamiltonian: 
 \begin{equation*}
\{H,I\}=\sum_{i=1}^n \left(\frac{\p H}{\p p_i}\frac{\p I}{\p x^i}-\frac{\p H}{\p x^i}\frac{\p I}{\p p_i}\right)=0.
 \end{equation*}

The natural action of the isometry group on the cotangent bundle $T^{*}N$ is 
Hamiltonian
and it preserves the energy $H$. Thus, the isometries represent infinitesimal symmetries
of the geodesic flow given, by virtue of Noether's theorem, by \textit{linear in 
momenta} integrals of motion. 
Explicitly, if $X = X^i(x)\p_{x^i}\in\mathfrak{iso}(N,g)$ is a Killing vector field 
then 
the corresponding integral is $I(x,p) =\langle X,p\rangle= X^i(x) p_i$. 

More generally, a \textit{Killing tensor} of degree $d$ corresponds to a 
\textit{homogeneous in momenta} polynomial
 \begin{equation}\label{Id}
I_d:= a^{i_1\cdots i_d}(x)\ p_{i_1}\cdots p_{i_d},
 \end{equation}
which Poisson commutes with $H$, and is thus a polynomial integral. Since the Hamiltonian is quadratic 
in momenta, for any \eqref{Id} the Poisson bracket $\{H,I_{d}\}$ is of degree $d+1$ in momenta. 
Consequently, Killing $d$-tensors correspond to solutions of a system of differential equations formed
by vanishing of $p$-coefficients of the Poisson bracket, which we call the \textit{Killing equation},
 \begin{equation}\label{geodesicflowPDE}
\mathcal{E}_d:= \{F = 0 : F \in \text{coeffs}_p(\{H, I_d\}) \}.
 \end{equation}
This is an overdetermined system of linear first order PDEs on the coefficients $a^{i_1\cdots i_d}(x)$ of the Killing tensor. 
Actually, there are ${n+d\choose d+1}$ equations on ${n+d-1\choose d}$ unknown functions.
Denote solutions to this system -- the linear space of all Killing $d$-tensors -- by $K_d$.

\subsection{Jet spaces and equations}\label{S22}

The notion of jet-space formalizes the computational device of truncated
Taylor polynomials; we refer for details to \cite{KL}.
If $x^i$ are local coordinates on $N$ then the jet-space $J^kN$ of $k$-jet of functions $u:N\to\R$ 
has local coordinates $(x^i,u_\sigma)$ for multi-indices $\sigma=(i_1,\dots,i_n)$,
$i_s\ge0$, $|\sigma|=\sum i_s\leq k$. Similarly are defined jets of vector-valued functions, sections, etc.
For a bundle $\pi:E\to N$ the jet-space of its sections is denoted by $J^k(N,E)$.

The space of $k$-jets of maps $u:\R^n\to\R^m$ will be simply denoted by $J^k(n,m)$.
It is a bundle of rank $m\cdot{n+k-1\choose k}$ over $n$-dimensional base.
Any map $u=(u^j):\R^n\to\R^m$ lifts to the jet-section
$j^ku:\R^n\to J^k(n,m)$ given by $x^i\mapsto u^j_\sigma=\partial u^j(x)/\p x^\sigma$.

 \begin{definition}[\textbf{Geometric PDE}]
A partial differential equation of order $k$ is 
a submanifold $\mathcal{E} \subseteq J^k(n,m)$. A 
\textit{solution} of the PDE is defined to be a function $u:\R^n\to\R^m$
such that its $k$-jet $j^ku$ takes values in $\mathcal{E}$.
A local solution is the same but defined on a domain $U\subset\R^n$.
We denote by $\op{Sol}(\mathcal{E})$ the space of all (local) solutions of 
$\mathcal{E}$.
 \end{definition} 

Elements of a $k$'th order geometric PDE $\mathcal{E} \subseteq 
J^k(n,m)$ are solutions \textit{up to order} $k$ (at a point).
To find the solutions of the PDE $\mathcal{E}$ 
up to order $k+1$ and higher, we have to differentiate the defining equations. 
To encode the chain rule, we define the $q$'th \textit{total derivative} of $F:J^k\to\R^s$
to be a vector-function on $J^{k+1}$ given by 
 \begin{equation}\label{eqn: total derivative}
D_qF:= \frac{\p F}{\p x^q}+ \sum_{j=1}^m \sum_{|\sigma|\leq k} 
\frac{\p F}{\p u^j_\sigma}\cdot u^j_{\sigma+1_q}.
 \end{equation}
(Here we use the notation $\sigma + 1_q$ for the multi-index obtained by adding 1
to the $q$'th entry of $\sigma$.) 
Now, a point $(x^i,u^j_\sigma)\in J^{k+1}$ is said to be a solution of $\mathcal{E}$ 
up to order $k+1$ if it satisfies the following system of equations:
 \begin{equation*}
\mathcal{E}^{(1)} := \Bigl\{ F(x^i,u^j_\sigma)= 0,\ (D_qF)(x^i,u^j_\alpha)= 0\ 
\forall\ q=1,\dots, n\Bigr\}.
 \end{equation*}
The resulting system of equations is called the \textit{first prolongation} of $\mathcal{E}$. 
By construction, a solution of the prolongation $\mathcal{E}^{(1)}$ is still a solution of $\mathcal{E}$. 
We inductively define the $l$'th prolongation by $\mathcal{E}^{(l)} = (\mathcal{E}^{(l-1)})^{(1)}\subset J^{k+l}$.
It corresponds to solutions \textit{up to order} $k+l$ (at a point).

 \begin{definition}[\textbf{Finite Type}]
A PDE $\mathcal{E}\subseteq J^k(n,m)$ is called of \textit{finite type} $l$ if after $l$ prolongations all 
the highest order derivatives of the dependent variables can be expressed algebraically in terms
of the lower order derivatives. A PDE is called of \textit{Frobenius type} if it is of finite type $0$.
 \end{definition}

Given a PDE $\mathcal{E}$ of finite type, it is readily seen that the space of formal 
solutions $\mathcal{E}^{(\infty)}$ is necessarily finite-dimensional. This implies that
 (under some regularity conditions) 
the solution space $\text{Sol}(\mathcal{E})$ is finite-dimensional. 

The Killing PDE is represented by a first order system $\mathcal{E}_d\subset J^1(N,S^dTN)$.
The following fundamental result is well-known, cf. \cite{To} and \cite{Wo}.

 \begin{theorem}[\textbf{Killing PDE is of Finite Type}]\label{thm:killing_finite_type}
The PDE $\mathcal{E}_d$ defining a Killing $d$-tensor is a first order linear 
PDE of finite type $d$ with $\op{Sol}(\mathcal{E}_d)=K_d$. This equation and its prolongations possess 
no compatibility conditions before achieving Frobenius type.
 \end{theorem}


\subsection{Prolongation-projection}\label{S23}

Let $\mathcal{E}=\{ F(x^i,u^j_\alpha)=0\}\subseteq J^k(n,m)$ be a PDE of order $k$. 
Its solution up to order $k$ can be extended to order $(k+l)$ if and only if it belongs 
to the projection of the prolongation $\pi_{k+l,k}(\mathcal{E}^{(l)})\subseteq \mathcal{E}$. 
In the case of equality here, every $k$-jet solution can be extended to a $(k+l)$-jet solution. 
In the opposite case, there is a linear combination of iterated total derivatives up to order $l$,
$\Box(F)=\sum_{|\tau|\leq l} a^\tau D_\tau F$, which has order $k$.

 \begin{definition}[\textbf{Compatibility}]
A \textit{compatibility condition} of $\mathcal{E}$ is an equation defining 
$\pi_{k+l,k}(\mathcal{E})$ that is algebraically independent of $F$ and 
that is satisfied by all formal solutions.
 \end{definition}

Associated to a PDE is the Cartan distribution. Solutions arise as integral manifolds of this distribution, cf.\ \cite{KL}. 
Therefore, in the PDE setting, Frobenius theorem implies:

 \begin{theorem}[\textbf{Frobenius Theorem}]\label{thm_frobenius2}
Solutions of a PDE $\mathcal{E} \subseteq J^k(n,m)$ of finite type $l$ are 
determined uniquely by their $(k+l-1)$-jets. If in addition $\mathcal{E}$ 
has no compatibility conditions, then for every $\xi\in\mathcal{E}^{(l)}$ there exists a local solution 
$u\in\op{Sol}(\mathcal{E})$ satisfying $j_x^{k+l}u=\xi$. 
 \end{theorem}

Note that if a PDE is of finite type, then all of its prolongations are finite type as well. There is the following 
more general claim, which holds also true in infinite type case, under the assumption of analyticity of the equation.

 \begin{theorem}[\textbf{Cartan's Involution}]\label{thm_cartan_prolongation}
There exists $q\in\mathbb{N}$ such that $\mathcal{E}^{(q)}$ is compatible. 
 \end{theorem}

Thus, in regular domains, there are only finitely many compatibility conditions.
However to find them explicitly is generally difficult, and bringing to involution in practice is a
formidable computation. We therefore substitute searching for involution by the following criterion.
For the finite type $l$ case:
{\it If $\pi:\mathcal{E}^{(r)}\to\mathcal{E}^{(r-1)}$ is surjective for some $r>l$, 
then $\mathcal{E}^{(r-1)}$ is compatible}.
This is especially simple for linear overdetermined PDEs: over regular domains $U\subset N$ such 
$\mathcal{E}$ are vector bundles and on each step of the prolongation-projection a compatibility 
condition reduces its rank; once this rank is stabilized for one step, then by Theorem 
\ref{thm_frobenius2} the system is compatible,
so the involution level $q$ of Theorem \ref{thm_cartan_prolongation} is achieved.

\subsection{Algorithmic implementation}\label{S24}

The above criterion allows for an effective implementation of evaluation of $\dim K_d$ for a given metric $g$ 
using computer algebra systems.

The Killing PDE $\mathcal{E}_d$ as well as its prolongations $\mathcal{E}_d^{(k)}$ are linear in
$(k+1)$-jets of the dependent variables. We convert this linear system of equations to a 
matrix-valued function $M_k(x)$ on the spacetime. For our class of metrics $g$ the entries are polynomials
with rational coefficients. Hence to make use of computer algebra software, we insert a {\it rational} point $x_0\in N$ 
to obtain a matrix with \textit{rational} coefficients (in this case computer calculations are exact!). 

The first thing to do is to find the points that work nicely with Cartan's prolongation-projection method. We call a point 
$x_0\in N$ \textit{regular} if the function $x\mapsto\op{rank}(M_k(x))$ attains its maximum at $x_0$ for all $k \geq 0$, 
that is, at each step we find the maximal number of compatibility conditions. Note that a regular point is generic, i.e., 
the set of regular points is an open dense subset of $N$. A point is \textit{singular} if it is not regular. 

   \begin{tcolorbox}
\textbf{Algorithm 1. (Cartan's Prolongation Method for Geodesic Flow)}.\newline
(\textbf{Input}: A nonnegative integer $d$, a regular point $x_0$.)
	\begin{itemize}
\item Step 1.) Compute the Poisson bracket $\{H, I_d\}$ of a polynomial in momenta $p$ function 
    $I_d$ with the Hamiltonian $H$. 
\item Step 2.) Collect the coefficients of $\{H, I_d\}$ with respect to 
    the momentum variables. Define the first order linear PDE $\mathcal{E}_d:= 
    \{F = 0:  F \in \text{coeffs}_{p}(\{H, I_d\})  \}$.
\item Step 3.) Set $k:= 0$.
    \begin{itemize}
\item Convert the linear system of equations $\mathcal{E}_d^{(k)}$ 
        w.r.t. the variables $\mathcal{V}_{k+1, d} := 
        \{a_{\alpha}^{i_1 \cdots i_d} : | \alpha | \leq k +1 \}$
        into a matrix $M_k(x)$ that depends on the $x$-coordinates. 
\item Substitute $x_0$ to obtain a matrix $M_k:= 
        M_k(x_0)$, the $k$'th \textit{prolongation matrix}. 
\item Set $\delta_k := \text{columns}(M_k) - \text{rank}(M_k)$.
    \end{itemize}
If $(k\leq d)$ or ($k>d$ and $\delta_k \neq \delta_{k-1}$), increase 
    $k$ by $1$ and repeat Step $3$.
\item Step 4.) Return $(\delta_k, k)$.
	\end{itemize}
(\textbf{Output}: The dimension of the space of Killing $d$-tensors is $\dim K_d = \delta_k$. 
The integer $k$ indicates the number of prolongations necessary to find all compatibility conditions 
of $\mathcal{E}$.)
   \end{tcolorbox}

 \begin{proposition}
Algorithm $1$ is correct and it terminates.
 \end{proposition}

 \begin{proof}
Termination is clear. We now justify correctness, i.e., that the algorithm computes 
the number of Killing $d$-tensors. Turn the prolongation matrix $M_k$ into row reduced echelon form. 
The equation $\mathcal{E}_d^{(k)}$ is linear and its rank as a bundle over $N$, equal to 
$\text{columns}(M_k)$, counts the number of $(k+1)$-jets of dependent variables. 
Rows of the matrix represent equations defining $\mathcal{E}_d^{(k)}$, so they consist of the 
original Killing PDE, their differential corollaries and compatibility conditions. Consequently,
$\delta_k$ is number of free jets (coordinates on fibers of the equation 
$\mathcal{E}_d\rightarrow N$). 
In view of the Frobenius theorem, each free variable corresponds to a $(k+1)$ jet-solution of the 
Killing PDE.

Now consider the conditions in step 3 determining termination of the loop.
The first part $(k \leq d)$ addresses whether the prolongation has achieved Frobenius type, see Theorem \ref{thm:killing_finite_type}. The second part ($k>d$ and $\delta_k\neq\delta_{k-1}$) checks
whether all compatibility conditions have been computed, as guaranteed by the criterion 
after Theorem \ref{thm_cartan_prolongation}. Thus every $(k+1)$ jet yields a local solution.
 \end{proof}

\subsection{Syzygies and Irreducible Killing Tensors}\label{S25}

The pointwise multiplication of functions gives rise to a linear map
   \begin{equation}
K_{d_1}\otimes K_{d_2}\mapsto K_{d_1+d_2},\ I_{d_1}\otimes I_{d_2}\mapsto I_{d_1}\cdot I_{d_2}.
   \end{equation}
A \textit{relation (syzygy)} among Killing tensors of rank $d_1$ and $d_2$ with $d_1\neq d_2$ 
is an element of the kernel of the map 
   \begin{equation}
K_{d_1} \otimes K_{d_2} \rightarrow K_{d_1 + d_2}.
   \end{equation}
If $d_1=d_2=:d$ a relation is given by an element in the kernel of the map $S^2K_d \rightarrow K_{2d}$.

A Killing $d$-tensor ($d \geq 2$) is \textit{irreducible} if it is not a linear combination of
the symmetric product of lower rank Killing tensors. The number of irreducible 
Killing $d$-tensors can be found using the number of syzygies. We demonstrate this for Killing $2$-tensors. 
The space of irreducible Killing $2$-tensors can be identified with the cokernel of map $\iota_2: S^2K_1 
\rightarrow K_2$, fitting into a short exact sequence
   \begin{equation}
0 \longrightarrow \text{Ker}\ \iota_2  \longrightarrow S^2K_1 \rightarrow K_2
\longrightarrow \text{Coker}\ \iota_2 \longrightarrow 0.
   \end{equation}
The space of irreducible Killing $3$-tensors can be identified with the cokernel of the map $\iota_3: K_1 \otimes K_2 
\rightarrow K_3$, etc. The number of syzygies among Killing tensors is found as 
follows. (We use the notation $\text{Taylor}(a(x), x_0, k)$ for the Taylor 
polynomial of the function $a$ around $x_0$ up to order $k$.)
   \begin{tcolorbox}
\textbf{Algorithm 2. (Finding Relations among Killing Tensors)}.\newline
(\textbf{Input}: Nonnegative integers $d_1, d_2$, a regular point $x_0$.)
	\begin{itemize}
\item Step 1.) For $s=1,2$: run algorithm 1 obtain $\dim K_{d_s}$ and 
    the number of prolongation $k_s$ needed to achieve compatibility.\\ 
    Consider the polynomial $I_{k_s+1, d_s} := \text{Taylor}(a^{i_1 \cdots i_{d_s}}, x_0, 
    k_s+1)\cdot p_{i_1} \cdots p_{i_{d_s}}$.
\item Step 2.) Consider the linear algebraic system of equations
    $\{\text{Taylor}(c, x_0, k_s) = 0: c \in \text{coeffs}_p(\{H, I_{k_s+1, d_s}\}) \}$ on
    the variables $\mathcal{V}_{k_s+1, d_s}(x_0)
    := \{a^{i_1 \cdots i_{d_s}}_{\alpha}(x_0) : |\alpha| \leq k_s + 1 \}$ for $s=1,2$. 
    Solve these linear equations and substitute the corresponding solutions into $I_{k_s +1,
    d_s}$ to obtain the truncated integrals $I_{k_s +1 , d_s}^j$ for $1 \leq j \leq \dim 
    K_{d_s}$.
\item Step 3.) Set \vskip-15pt
    \begin{equation*}
T:= \sum_{l_1=1}^{\dim K_{d_1}}\sum_{l_2=1}^{\dim K_{d_2}} c_{l_1,l_2}\ I_{k_1+1,d_1}^{l_1}\cdot I_{k_2+1,d_2}^{l_2}.
    \end{equation*}
    Define $S :=\{\text{Taylor}(c, x_0, d_1+d_2) : c \in\text{coeffs}_p(T)\})$.
\item Step 4.) Solve the linear algebraic system of equations $ \{F = 0 : F\in\text{coeffs}_{x}(S) \}$ in terms 
of the coefficients $c_{l_1,l_2}$, and denote the resulting solution space $R$. 
\item Step 5.) Return $R$ and $\dim R$.
	   \end{itemize}
(\textbf{Output}: Relations among Killing tensors of rank $d_1$ and $d_2$;
$\#$ (indep) syzygies $=\dim R$.)
   \end{tcolorbox}

   \begin{proposition}
Algorithm 2 is correct and it terminates.
   \end{proposition}

   \begin{proof}
For $d \geq 1$, consider the Killing PDE $\mathcal{E}_d\subseteq J^1$. A $(k+1)$-jet $j_{x_0}^{k+1}u$ 
of a vector-function $u= (a^{i_1\cdots i_d}(x))$ can be identified with the Taylor polynomial 
$I_{k+1,d} =\text{Taylor}(a^{i_1 \cdots i_d}, x_0, k+1) p_{i_1} \cdots p_{i_d}$. 
Under this correspondence, we have that $j_{x_0}^{k+1}u\in\mathcal{E}^{(k)}$ if and only if 
$\{H,I_{k+1,d}\}$ vanishes up to order $k$ at $x_0$. These observations explain steps 1 and 2.

By \Cref{thm:killing_finite_type} $K_d$ is determined by $d$-jets in the sense
that we can compute all jets of a Killing $d$-tensor at a point if we know its $d$-jet. 
Thus, in step $3$ we must include the jets up to order $d_1 + d_2$ in order to determine uniquely 
the corresponding $(d_1 + d_2)$-tensor.
   \end{proof}

\textbf{Application.} In practice we apply algorithm 2 as follows.
First, using algorithm $1$ we compute the dimensions of $S^2K_1$ and $K_2$. Then we 
use algorithm $2$ to determine the dimension of the kernel $\text{Ker}\ \iota_2$. 
Finally, the number of (lin.\ independent) irreducible Killing $2$-tensors is given by
  \begin{equation*}
\dim \text{Coker}\ \iota_2 = \dim K_2 - \dim S^2K_1 + \text{Ker}\ \iota_2.
  \end{equation*}
This method can be readily generalized to higher order Killing tensors.

\medskip

\textbf{Regular and singular points.}
Even though the regular points are dense, it is difficult to verify (in practice) that
a given point is regular. Thus, we must be careful in order to get rigorous results. 
For a singular point, algorithm 1 gives an upper 
bound on the number of Killing tensors. The number of syzygies imply
lower bounds on the number of Killing tensors (indeed, the syzygies 
imply the number of reducible Killing tensors). Thus, whenever the 
algorithms suggest the existence of an irreducible Killing tensor it is
important to find it explicitly. (For our metrics $g$ it turns out to be   
possible to find the irreducible Killing tensors explicitly using Maple's pdsolve.)

\newpage 
\subsection{Note on the computability of the algorithm.}\label{S26}

We briefly discuss the computational difficulties associated with the proposed method and how we 
deal with them. Dimension of the prolongation matrix $M_k$ from algorithm 1 equals
   \begin{equation*}
\text{rows}(M_k) = {n+d \choose d+1}\cdot {n+k\choose n},\quad
\text{columns}(M_k) = {n+d-1\choose d}\cdot {n+k+1\choose n}.
   \end{equation*}
In particular, we see that the number of rows grows faster with $k$ than the number of columns.
We highlight several elements that have made the computer implementation 
more efficient:

   \begin{itemize}[leftmargin=*]
\item (\textbf{LinBox}). The LinBox package \cite{LinBox} in Sage allows for incredibly fast rank computations 
of large sparse integer matrices. For example, computing the rank of the quartic 
prolongation matrix $M_{19}$ for metric 2 with size $(495880) \times (371910) $ took less than an hour. 
In comparison, rank computations of smaller matrices (say 50000 by 40000) would take several days 
in Maple or not give a result at all. Thanks to LinBox, the time to compute the ranks is 
negligible. Generating a 
prolongation matrix takes by far the longest time of the steps in algorithm 1.

\medskip

\item (\textbf{Exploiting Sparsity}.) The prolongation matrices $M_k$ that we encounter here are sparse 
(with density $<$ 0.001). It is important that the generation of the matrix reflects this. We generate 
the initial matrix with all entries zeroes and then substitute the nonzero values. 

\medskip

\item (\textbf{Combinatorial Description of Prolongations}.)
For the quartic case, we used a combinatorial description of the prolongation equations. 
We demonstrate this for metric $2$. Since $I_4$ is of degree 4, we have that $\{H, I_4\}$ is of degree 5 in momenta. 
Thus, we can write $\{H, I_4\} = \text{coeff}_{\tau} p^{\tau}$ where $p^{\tau} = 
p_1^{\tau_1}p_2^{\tau_2}p_3^{\tau_3}p_4^{\tau_4}$. Given a multi-index $\tau$ of 
length $5$, we obtain the $p^{\tau}$-coefficient in terms of the coefficients of $I$:
    \begin{equation*}
\begin{split}
\hphantom{aaa}\text{coeff}_{\tau}(\{H, I_4\}) & =  2 \partial_1(a^{\tau - 1_1})  + 2 
            \partial_2(a^{\tau - 1_2}) - 2 \partial_4(a^{\tau - 1_3}) - 
            2\partial_3(a^{\tau - 1_4}) \\
           & + 4x^3((x^1)^2 + (x^2)^2) \partial_4(a^{\tau - 1_4}) - 2((x^1)^2 + 
           (x^2)^2) \frac{(\tau + 1_3 - 2 \cdot 1_4)!}{(\tau - 2\cdot 1_4)} \\
            & -4 x^1 x^3  \frac{(\tau +1_1  - 2 \cdot 1_4)!}{(\tau - 2\cdot 1_4)!} 
            a^{\tau + 1_1 - 2\cdot 1_4} -4 x^2 x^3 \frac{(\tau +1_2 - 2 \cdot 
            1_4)!}{(\tau - 2\cdot 1_4)!} a^{\tau + 1_2 - 2\cdot 1_4}.
\end{split}
    \end{equation*}
Using the Leibniz rule for multi-index notation, we can subsequently determine the 
general expression for the derivative $\partial^{\alpha}(\text{coeff}_{\tau}(\{H, I_4\}))$, where 
$\alpha$ is a multi-index. In this way we obtain the equations of the prolongation 
as a function of the multi-indices $\tau$ and $\alpha$. This combinatorial 
description significantly reduces the time needed to generate the equations in 
Maple, especially as the order increases. This approach is most beneficial 
for Hamiltonians which are polynomials of low order in the independent $x$-variables. 
(For the Kerr metric, for example, these combinatorics would be unfeasible.) 
   \end{itemize}

\section{Computations and results} \label{S3}

Here we discuss the results of concrete computations with the above algorithms.
We begin with investigations of the special cases of pp-waves and Wils metric and 
then discuss the general case.

\subsection{Conformally Flat pp-Waves}\label{S31}
These are given by the following formula:
   \begin{equation}\label{ppw}
g= 2dx^3dx^4 + \bigl(f(x^3)((x^1)^2+(x^2)^2)\bigr)(dx^3)^2  -(dx^1)^2 -(dx^2)^2
   \end{equation}

Sippel and Goenner classified pp-waves in terms of their isometry groups \cite{SG}. 
For conformally flat pp-waves there are three classes: $f(x^3)=c$, $f(x^3)=c(x^3)^{-2}$ 
and the generic case with $\dim K_1=6$. We apply our prolongation-projection algorithm
to the following four metrics (rescaling of $f$ does not play a role for the first three metrics): 
 $$
{\rm(i)}\ f(x^3)=1,\quad {\rm(ii)}\ f(x^3)=x^3,\quad {\rm(iii)}\ f(x^3)=(x^3)^{2},\quad {\rm(iv)}\ f(x^3)=2 (x^3)^{-2}.
 $$ 
If two subsequent values $\delta_k$, $\delta_{k+1}$ 
are equal (with $k \geq d$), the sequence of $\delta$-values stabilizes 
and we can read off the number of Killing $d$-tensors. In the table this is shown by circling 
this $\delta$-value. 

\begin{table}[h!]
{\rowcolors{1}{red!30!green!30!blue!20!}{}
\begin{tabular}{||c | c | c |  c  | c | c ||} 
    \hline
    Linear\hspace{0.56cm} & $\mathcal{E}$ & $\mathcal{E}^{(1)}$ & 
    $\mathcal{E}^{(2)}$ & $\mathcal{E}^{(3)}$ & $\mathcal{E}^{(3)}$ \\ [0.5ex]
    \hline
    $\delta$ & 10 & 10  & 7  & \circled{7} & \dots \\ 
    \hline
    Quadratic & $\mathcal{E}$ & \dots & $\mathcal{E}^{(4)}$ 
    & $\mathcal{E}^{(5)}$ & $\mathcal{E}^{(6)}$ \\ [0.5ex] 
    \hline
    $\delta$ & 30 & \dots & 29 & 28 & \circled{28} \\ 
    \hline
    Cubic\hspace{0.64cm} & $\mathcal{E}$ & $\dots$ & $\mathcal{E}^{(6)}$ & 
    $\mathcal{E}^{(7)}$ & $\mathcal{E}^{(8)}$ \\[0.5ex]
    \hline
    $\delta$ & 65 & \dots & 87 & 84 & \circled{84} \\
    \hline
    Quartic\hspace{0.39cm} & $\mathcal{E}$ & \dots & $\mathcal{E}^{(10)}$ & 
    $\mathcal{E}^{(11)}$ & $\mathcal{E}^{(12)}$ \\[0.5ex] 
    \hline
    $\delta$ & 119 & \dots & 211  & 210 & \circled{210}  \\ 
    \hline
\end{tabular}}
{\rowcolors{1}{red!30!green!30!blue!20!}{}
\begin{tabular}{||c | c | c |  c  | c | c ||} 
    \hline
    Linear\hspace{0.56cm} & $\mathcal{E}$ & $\mathcal{E}^{(1)}$ & 
    $\mathcal{E}^{(2)}$ & $\mathcal{E}^{(3)}$ & $\mathcal{E}^{(3)}$ \\ [0.5ex]
    \hline
    $\delta$ & 10 & 10  & 7  & 6 & \circled{6} \\ 
    \hline
    Quadratic & $\mathcal{E}$ & \dots & $\mathcal{E}^{(5)}$ 
    & $\mathcal{E}^{(6)}$ & $\mathcal{E}^{(7)}$ \\ [0.5ex] 
    \hline
    $\delta$ & 30 & \dots & 24 & 22 & \circled{22} \\ 
    \hline
    Cubic\hspace{0.64cm} & $\mathcal{E}$ & $\dots$ & $\mathcal{E}^{(11)}$ & 
    $\mathcal{E}^{(12)}$ & $\mathcal{E}^{(13)}$ \\[0.5ex]
    \hline
    $\delta$ & 65 & \dots & 63 & 62 & \circled{62} \\
    \hline
    Quartic\hspace{0.39cm} & $\mathcal{E}$ & \dots & $\mathcal{E}^{(17)}$ & 
    $\mathcal{E}^{(18)}$ & $\mathcal{E}^{(19)}$ \\[0.5ex] 
    \hline
    $\delta$ & 119 & \dots & 150   & 148   & \circled{148}  \\ 
    \hline
\end{tabular}}
\end{table}

\begin{table}[h!]
{\rowcolors{1}{red!30!green!30!blue!20!}{}
\begin{tabular}{||c | c | c |  c  | c | c ||} 
    \hline
    Linear\hspace{0.56cm} & $\mathcal{E}$ & $\mathcal{E}^{(1)}$ & 
    $\mathcal{E}^{(2)}$ & $\mathcal{E}^{(3)}$ & $\mathcal{E}^{(3)}$ \\ [0.5ex]
    \hline
    $\delta$ & 10 & 10  & 7  & 6 & \circled{6} \\ 
    \hline
    Quadratic & $\mathcal{E}$ & \dots & $\mathcal{E}^{(5)}$ 
    & $\mathcal{E}^{(6)}$ & $\mathcal{E}^{(7)}$ \\ [0.5ex] 
    \hline
    $\delta$ & 30 & \dots & 24 & 22 & \circled{22} \\ 
    \hline
    Cubic\hspace{0.64cm} & $\mathcal{E}$ & $\dots$ & $\mathcal{E}^{(11)}$ & 
    $\mathcal{E}^{(12)}$ & $\mathcal{E}^{(13)}$ \\[0.5ex]
    \hline
    $\delta$ & 65 & \dots & 63 & 62 & \circled{62} \\
    \hline
    Quartic\hspace{0.39cm} & $\mathcal{E}$ & \dots & $\mathcal{E}^{(17)}$ & 
    $\mathcal{E}^{(18)}$ & $\mathcal{E}^{(19)}$ \\[0.5ex] 
    \hline
    $\delta$ & 119 & \dots & 150   & 148   & \circled{148}  \\ 
    \hline
\end{tabular}}
{\rowcolors{1}{red!30!green!30!blue!20!}{}
\begin{tabular}{||c | c | c |  c  | c | c ||} 
    \hline
    Linear\hspace{0.56cm} & $\mathcal{E}$ & $\mathcal{E}^{(1)}$ & 
    $\mathcal{E}^{(2)}$ & $\mathcal{E}^{(3)}$ & $\mathcal{E}^{(3)}$ \\ [0.5ex]
    \hline
    $\delta$ & 10 & 10  & 7  & \circled{7} & \dots \\ 
    \hline
    Quadratic & $\mathcal{E}$ & \dots & $\mathcal{E}^{(4)}$ 
    & $\mathcal{E}^{(5)}$ & $\mathcal{E}^{(6)}$ \\ [0.5ex] 
    \hline
    $\delta$ & 30 & \dots & 29 & 28 & \circled{28} \\ 
    \hline
    Cubic\hspace{0.64cm} & $\mathcal{E}$ & $\dots$ & $\mathcal{E}^{(6)}$ & 
    $\mathcal{E}^{(7)}$ & $\mathcal{E}^{(8)}$ \\[0.5ex]
    \hline
    $\delta$ & 65 & \dots & 87 & 84 & \circled{84} \\
    \hline
    Quartic\hspace{0.39cm} & $\mathcal{E}$ & \dots & $\mathcal{E}^{(10)}$ & 
    $\mathcal{E}^{(11)}$ & $\mathcal{E}^{(12)}$ \\[0.5ex] 
    \hline
    $\delta$ & 119 & \dots & 211  & 210 & \circled{210}  \\ 
    \hline
\end{tabular}}
\caption{Left up: Metric (i) $f(x^3) = 1$;\hspace{1.5cm} Right up: Metric (ii) $f(x^3)=x^3$.\newline
Left down: Metric (iii) $f(x^3)=(x^3)^{2}$;\hspace{1.95cm} Right down: Metric (iv) 
$f(x^3)=2(x^3)^{-2}$.}
\end{table}

We see that metrics 1 and 4 have 7-dimensional isometry, which is 
consistent with the classification by Sippel and Goenner. Note that for
the quartic case of metrics 2 and 3 we have to go all the way to the 
19'th prolongation of the Killing PDE. The number of equations and 
variables at this stage are so large that it is unlikely that we can 
compute the number of Killing $5$-tensors with present computational powers.

\subsection{Syzygies and irreducible Killing tensors for pp-waves}\label{S32}

In order to find the number of irreducible Killing tensors, we have 
to take into account the number of syzygies among the Killing 
tensors. We demonstrate this for metric 2 (the other cases are similar). 
Consider the following short exact sequence
\begin{equation*}
    0 \longrightarrow \underbrace{\text{Ker}\ \iota_2}_{\text{1 
    syzygy}} \rightarrow 
    \underbrace{S^2K_1}_{21-\text{dim.}} 
    \xrightarrow{\iota_2} \underbrace{K_2}_{22-\text{dim}.} \longrightarrow 
    \underbrace{\text{Coker}\ \iota_2}_{\text{2 irreducible Killing 
    $2$-tensors}} \rightarrow 0
\end{equation*}
Algorithm 2 gives $\dim \text{Ker}\ \iota_2 = 1$, and so there are 2 
irreducible Killing 2-tensors. Next, we consider
\begin{equation*}
    0 \longrightarrow \underbrace{\text{Ker}\ \iota_3}_{\text{70 syzygies}} \rightarrow 
    \underbrace{K_1 \otimes K_2}_{132-\text{dim.}} 
    \xrightarrow{\iota_3} \underbrace{K_3}_{62-\text{dim}.} \longrightarrow 
    \underbrace{\text{Coker}\ \iota_3}_{\text{0 irreducible Killing 
    $3$-tensors}} \rightarrow 0\end{equation*}
Algorithm $2$ gives $\dim \text{Ker}\ \iota_3 = 70$ and so there are no
irreducible Killing $3$-tensors. Since there are no irreducible Killing $3$-tensors, the source space 
of $\iota_4$ is the second symmetric power $S^2K_2$. (If there were irreducible Killing $3$-tensors, 
then the source space would be $K_1\otimes K_3\oplus S^2K_2$.) Thus we obtain the short exact sequence
\begin{equation*}
    0 \longrightarrow \underbrace{\text{Ker}\ \iota_2}_{\text{105 syzygies}} \rightarrow 
    \underbrace{S^2K_2}_{253-\text{dim.}} 
    \xrightarrow{\iota_4} \underbrace{K_4}_{148-\text{dim}.} \longrightarrow 
    \underbrace{\text{Coker}\ \iota_4}_{\text{0 irreducible Killing 
    $4$-tensors}} \rightarrow 0
\end{equation*}
There are 105 syzygies, it follows that there are no irreducible Killing $4$-tensors.

For metrics ${\rm(i)}$, ${\rm(ii)}$ and ${\rm(iii)}$ we obtain that there exists one irreducible Killing $2$-tensor, 
in addition to the Hamiltonian. Actually, we can explicitly write this Killing 2-tensor as follows:
   \begin{equation}\label{qi}
J:= -x^3 H + x^1 p_1 p_4 + x^2 p_2 p_4 + 2 x^4 p_4^2.
   \end{equation}
For metric ${\rm(iv)}$ the only irreducible Killing $2$-tensor is the Hamiltonian $H$, i.e., the Killing 
2-tensor $J$ is reducible in this case (due to the existence of an extra Killing vector). 

In the general case \eqref{ppw} for $f(u)\neq c,cu^{-2}$ the Killing vectors are the following:
 \begin{equation}\label{KVppw}
I_1= p_1x^2-p_2x^1,\ I_2=p_4,\ I_{3,4}=a_{1,2}(x^3)p_1+a_{1,2}'(x^3)x^1p_4,\
I_{5,6}=a_{1,2}(x^3)p_2+a_{1,2}'(x^3)x^2p_4,
 \end{equation}
where $a_i$ ($i=1,2$) are fundamental solutions of the linear second order ODE $a''+fa=0$, i.e.,
solutions satisfying the initial conditions $a_1(0)=1,a'_1(0)=0$, $a_2(0)=0,a'_2(0)=1$.
The Hamiltonian is equal to
 \begin{equation}\label{Hppw} 
H=2p_3p_4-p_1^2-p_2^2-((x^1)^2+(x^2)^2)f(x^3)p_4^2
 \end{equation}
and the other quadratic integral $J$ is given by \eqref{qi} (also for general $f$).
These results are consistent with the following theorem by Keane and Tupper \cite{KT} that 
was proven using the Koutras algorithm \cite{Kou} (our approach is different).

   \begin{theorem}[{\cite{KT}}]
A conformally flat pp-wave with $\dim K_1=6$ or with $f(x^3)=c$ admits an irreducible Killing 
2-tensor, independent of the (irreducible) Hamiltonian $H$. 
   \end{theorem}

By using our computational algorithm we can also establish the nonexistence results of higher 
order Killing tensors for these conformally flat pp-waves.

   \begin{theorem}
A conformally flat pp-wave \eqref{ppw} with $f(u)=cu^m$, $m=0,1,2$, or $f(u) = 2u^{-2}$, admits no 
irreducible Killing $3$- and $4$-tensors.
   \end{theorem}

  \begin{proof}
The result follows straightforwardly from the above computations and a rescaling argument. 
  \end{proof}

   \begin{corollary}
For a generic conformally flat pp-wave \eqref{ppw} all $3$- and $4$- Killing tensors are combinations 
of Killing vectors \eqref{KVppw}, the Hamiltonian $H$ \eqref{Hppw} and the Killing 2-tensor $J$ \eqref{qi}.
   \end{corollary}

Here $f$ is generic in $C^{13}$ topology for Killing 3-tensors and in $C^{19}$ topology for Killing 4-tensors,
see Table 1 for $k=k_d$, however we believe that also holds in lower regularity by the approach of \cite{KM2}.
   
  \begin{proof}
It follows from our computations and algebraic dependence of the matrix $M_k$ on $j^{k+1}f$ that
$\dim K_i$ ($i=2,3,4$) is upper semi-continuous in this jet. Hence, for a generic $f(x^3)$ 
the dimension of $K_2,K_3,K_4$ are as indicated in the third term of the above short exact 
sequences. Due to full control of $K_1,K_2$ the second terms have dimensions as indicated. 
Dimension of the first term
is also upper semi-continuous, so for a generic $f(x^3)$ we have at most the indicated number of 
syzygies. In fact, this number is realizable as follows. 

In the case of Killing 2-tensor (first short exact sequence) the only syzygy is
(verifying this exploits constancy of the Wronskian of $a_1,a_2$)
 $$
\mathfrak{S}_2:\quad I_1I_2+I_3I_6-I_4I_5=0.
 $$
For Killing 3-tensor (second short exact sequence) the only 6 syzygies are $I_j\cdot\mathfrak{S}_2$ ($1\leq j\leq6$).
To explain dimension 70 of the first term, note that kernel of the symmetrization operator
$K_1\otimes S^2K_1/K_1\otimes\mathfrak{S}_2\to S^3K_1$ is 64-dimensional.
Similarly one justifies the case of Killing 4-tensor (third short exact sequence).

Actually, we can also obtain the claim from the fact that the functional rank of 8 functions $I_j,H,J$ is 7, while that
of $I_j$ is 5. Thus no syzygies can involve $H,J$ and the only syzygy among 6 Killing vectors $I_i$ is given by
$\mathfrak{S}_2$.
  \end{proof} 

\subsection{Absence of Killing Tensors for the Wils Metric}\label{S33}

The Wils metric is given by
   \begin{equation}\label{Wm}
g = 2x^1dx^3dx^4 -2x^4dx^1dx^3 +\bigl(f(x^3)x^1((x^1)^2+(x^2)^2)-(x^4)^2\bigr)(dx^3)^2  -(dx^1)^2-(dx^2)^2.
   \end{equation}

We apply our prolongation-projection algorithm to the following three cases: 
$f(u)=u^m$, $m=0,1,2$. The results are displayed in the following table.

\begin{table}[h!]
{\rowcolors{1}{red!30!green!30!blue!20!}{}
\begin{tabular}{||c | c | c |  c ||} 
    \hline
    Linear\hspace{0.56cm} & \dots & $\mathcal{E}^{(4)}$ & 
    $\mathcal{E}^{(5)}$  \\ [0.5ex]
    \hline
    $\delta$ & \dots & 1 & \circled{1} \\ 
    \hline
    Quadratic & \dots & $\mathcal{E}^{(5)}$ 
    & $\mathcal{E}^{(6)}$ \\ [0.5ex] 
    \hline
    $\delta$ & \dots & 2 & \circled{2} \\ 
    \hline
    Cubic\hspace{0.64cm} & $\dots$ & $\mathcal{E}^{(7)}$ & 
    $\mathcal{E}^{(8)}$  \\[0.5ex]
    \hline
    $\delta$ & \dots & 2 & \circled{2} \\
    \hline
    Quartic\hspace{0.39cm} & \dots & $\mathcal{E}^{(8)}$ & 
    $\mathcal{E}^{(9)}$  \\[0.5ex] 
    \hline
    $\delta$ & \dots & 3 & \circled{3}  \\ 
    \hline
    Quintic \hspace{0.39cm} & \dots & $\mathcal{E}^{(9)}$ & 
    $\mathcal{E}^{(10)}$  \\[0.5ex] 
    \hline
    $\delta$ & \dots & 3 & \circled{3}  \\ 
    \hline
    Sextic \hspace{0.39cm} & \dots & $\mathcal{E}^{(10)}$ & 
    $\mathcal{E}^{(11)}$  \\[0.5ex] 
    \hline
    $\delta$ & \dots & 4 & \circled{4}  \\ 
    \hline
\end{tabular}}
{\rowcolors{1}{red!30!green!30!blue!20!}{}
\begin{tabular}{||c | c | c |  c ||} 
    \hline
    Linear\hspace{0.56cm} & \dots & $\mathcal{E}^{(5)}$ & 
    $\mathcal{E}^{(6)}$  \\ [0.5ex]
    \hline
    $\delta$ & \dots & 0 & \circled{0} \\ 
    \hline
    Quadratic & \dots & $\mathcal{E}^{(5)}$ 
    & $\mathcal{E}^{(6)}$ \\ [0.5ex] 
    \hline
    $\delta$ & \dots & 1 & \circled{1} \\ 
    \hline
    Cubic\hspace{0.64cm} & $\dots$ & $\mathcal{E}^{(7)}$ & 
    $\mathcal{E}^{(8)}$  \\[0.5ex]
    \hline
    $\delta$ & \dots & 0 & \circled{0} \\
    \hline
    Quartic\hspace{0.39cm} & \dots & $\mathcal{E}^{(8)}$ & 
    $\mathcal{E}^{(9)}$  \\[0.5ex] 
    \hline
    $\delta$ & \dots & 1 & \circled{1}  \\ 
    \hline
    Quintic \hspace{0.39cm} & \dots & $\mathcal{E}^{(9)}$ & 
    $\mathcal{E}^{(10)}$  \\[0.5ex] 
    \hline
    $\delta$ & \dots & 0 & \circled{0}  \\ 
    \hline
    Sextic \hspace{0.39cm} & \dots & $\mathcal{E}^{(10)}$ & 
    $\mathcal{E}^{(11)}$  \\[0.5ex] 
    \hline
    $\delta$ & \dots & 1 & \circled{1}  \\ 
    \hline
\end{tabular}}
{\rowcolors{1}{red!30!green!30!blue!20!}{}
\begin{tabular}{||c | c | c |  c ||} 
    \hline
    Linear\hspace{0.56cm} & \dots & $\mathcal{E}^{(5)}$ & 
    $\mathcal{E}^{(6)}$  \\ [0.5ex]
    \hline
    $\delta$ & \dots & 0 & \circled{0} \\ 
    \hline
    Quadratic & \dots & $\mathcal{E}^{(5)}$ 
    & $\mathcal{E}^{(6)}$ \\ [0.5ex] 
    \hline
    $\delta$ & \dots & 1 & \circled{1} \\ 
    \hline
    Cubic\hspace{0.64cm} & $\dots$ & $\mathcal{E}^{(7)}$ & 
    $\mathcal{E}^{(8)}$  \\[0.5ex]
    \hline
    $\delta$ & \dots & 0 & \circled{0} \\
    \hline
    Quartic\hspace{0.39cm} & \dots & $\mathcal{E}^{(8)}$ & 
    $\mathcal{E}^{(9)}$  \\[0.5ex] 
    \hline
    $\delta$ & \dots & 1 & \circled{1}  \\ 
    \hline
    Quintic \hspace{0.39cm} & \dots & $\mathcal{E}^{(9)}$ & 
    $\mathcal{E}^{(10)}$  \\[0.5ex] 
    \hline
    $\delta$ & \dots & 0 & \circled{0}  \\ 
    \hline
    Sextic \hspace{0.39cm} & \dots & $\mathcal{E}^{(10)}$ & 
    $\mathcal{E}^{(11)}$  \\[0.5ex] 
    \hline
    $\delta$ & \dots & 1 & \circled{1}  \\ 
    \hline
\end{tabular}}
\caption{Metric (i) $f(x^3)=1$; \hspace{0.8cm} Metric (ii) $f(x^3)=x^3$;\hspace{2.1cm} Metric (iii) $f(x^3)=(x^3)^2$.}
\end{table}

 \begin{theorem}
The Wils metric \eqref{Wm} for $f(u)=u^m$, $m=0,1,2$, admits no Killing tensors up to degree 6 except for powers of
the Hamiltonian.
 \end{theorem}

This statement follows directly from Table 2.
It also implies that for generic values of the functional parameter $f$ there are no lower degree Killing tensors.
Now we want to be more specific on those exceptional parameters.

   \begin{theorem}\label{WW}
The Wils metric admits Killing vectors if and only if $f$ is of the form
   \begin{equation}\label{Wf}
f(x^3) = (c_0+c_1x^3+c_2(x^3)^2)^{-2}.
   \end{equation}
In this case the Killing vector is unique up to scale and is given by the formula
   \begin{equation}\label{WKV}
X:= (c_0+c_1x^3+c_2(x^3)^2)\,\p_{x^3}-(2c_2x^1+c_1x^4+2c_2x^3 x^4)\ \partial_{x^4}. 
\end{equation}
   \end{theorem}
   
   \begin{proof}
In order to simplify the calculations we evaluate at $x^1=1, x^2=2, x^4=4$ but leave $x^3$ general. 

\textit{Step 1 and 2.)} Using, the equations defining the PDE $\mathcal{E}$, we 
express the $1$-jets $a^1_1$, $a^1_2$, $a^1_3$, $a^1_4$, $a^2_3$, $a^2_4$, 
$a^3_2$, $a^3_4$, $a^4_1$, $a^4_2$ in terms of the free variables $a^1$, $a^2$, 
$a^3$, $a^4$, $a^2_1$, $a^2_2$, $a^3_1$, $a^3_3$, $a^4_3$, $a^4_4$ and the function $f(x^3)$. 

\textit{Step 3.)} For the first prolongation $\mathcal{E}^{(1)}$, we can express
all 2-jets in terms of lower order jets without making any assumptions on $f$. 

\textit{Step 4.)} Consider $\mathcal{E}^{(2)}$. If we assume that $f \neq 0$, we
obtain the following compatibility conditions:
 $$
a^3_1=0,\ a^3_3=-\frac{a^1f+f'a^3}{2f},\ a^4_4=0.
 $$
We are left with 7 free jet variables. For $\mathcal{E}^{(3)}$, we obtain the 
additional compatibility conditions:
 $$
a^1=0,\ a^2_2=a^4_3,\ a^2_1=\frac{2a^2f^2-4a^3ff'+2a^3ff''-3a^3(f')^2}{6f^2}. 
 $$
We are left with 4 free variables. The prolongation $\mathcal{E}^{(4)}$ gives 
three additional compatibility conditions: $a^4_{3} =0$ and two expressions for 
$a^2$ and $a^4$. Only 1 free variable $a^3$ remains, and the next prolongation $\mathcal{E}^{(5)}$ 
does not give an additional compatibility condition if and only if $f$ is a solution of the ODE
  \begin{equation}\label{eqn_ODE_Wils}
f''' = \frac{18ff'f''-15(f')^3}{4f^2}. 
  \end{equation}
Resolving this ODE gives the required formula \eqref{Wf}. Expression \eqref{WKV} 
follows.
   \end{proof}

The following theorem is proven in the same manner, but the number of steps is 
larger, so the proof is omitted.

   \begin{theorem}
The Wils metric has Killing 2-tensors if and only if it has nontrivial Killing vectors. This happens only for the functional
parameter \eqref{Wf}; in this case, denoting $I_1=\langle X,p\rangle$ the linear integral corresponding to \eqref{WKV}, 
the general quadratic integral is a linear combination $k_1I_1^2+k_2H$. 
   \end{theorem}

\subsection{General Koutras-McIntosh metrics}\label{S34}

Investigation of the general metric \eqref{KMm} follows the same scheme.
First of all, the computation in the previous section implies that the matrix $M_k$ of the prolonged
Killing PDE 
for degree $d\leq6$ tensors has minimal possible value for $\delta_k$, i.e., 0 for odd $d$ and $1$ 
for even $d$. This implies Theorem \ref{Th1}.

To obtains Theorems \ref{Th2} and \ref{Th3} we can perform general computation with symbolic matrix
for the prolongation $\mathcal{E}^{(6)}$ when $d=1$ and $\mathcal{E}^{(7)}$ when $d=2$.
The matrix $M_k$ has size $1260\times840$ for $d=1$ and $4200\times3300$ for $d=2$.
To compute its rank we use the idea exploited in \cite{KV}, namely successively identifying rows or columns 
with few non-zero terms (this means $\leq2$ for $d=1$ and $\leq8$ for $d=2$) and doing Gauss 
elimination, while storing the involved factors to check their vanishing separately. 
This gives the splitting $a=0$ or $b=0$ and the rest follows from Theorem \ref{WW}.
In fact, this computation also yields equation \eqref{eqn_ODE_Wils}.

The exceptional functional parameters $f(u)$ in \eqref{Wf} up to transformations $u\to ku+b$
(change of coordinates: $x^1\mapsto\lambda x^1$, $x^2\mapsto\lambda x^2$, $x^3\mapsto kx^3+b$, 
$x^4\mapsto\lambda x^4/k$, $g\mapsto\lambda^2 g$, $f\mapsto f/(\lambda k^2)$)
give the following different cases
 $$
f(u)=1,\quad  f(u)=u^{-1},\quad  f(u)=c u^{-2},\quad  f(u)=u^{-4},\quad f(u)=|u^2\pm1|^{-2}.
 $$
In each of these cases one can directly verify there are no irreducible Killing 3- or 4-tensors
(for the middle case this was only verified for a generic parameter $c$),
 i.e.\ all of them are algebraic combinations of $I_1$ and $H$.

\section{Outlook}

In this paper we obtain the nonexistence of Killing tensors of degrees $d$ up to 6 for the Koutras-McIntosh spacetimes
for generic parameters. 
This complements the previous result on the nonexistence of Killing vectors \cite{KoMc}.
The problem of existence of higher order $d>6$ Killing tensors remains open. 
The size of the involved matrices ($163800 \times 152880$ for $d=6$) 
does not allow further computational progress, and we have to stress that our success for metrics \eqref{KMm}
is related to sparsity of the corresponding matrices $M_k$ and rationality of their entries in coordinates and parameters.

The complexity of computations carried here is much higher than that in
preceeding works \cite{KM1,KV,V}; actually those possessed Killing vectors allowing to reduce
the PDE setup to that on a 2-dimensional manifold, while our setup here is fully 4-dimensional
(that is why the size of the matrix $M_k$ of $\mathcal{E}_d^{(k)}$ grows much faster).
Other works \cite{HM,KT,Ki}, addressing Killing 2-tensors, have in similar vein reductions to ODEs
(that is, differential equations on a 1-dimensional manifold), so our work on higher degree $d$ 
Killing tensors is apparently novel.

One may envision that the following approach is feasible for large $d$. Consider the Killing PDE $\mathcal{E}_d$
with $\op{Sol}(\mathcal{E}_d)=K_d$. This is an overdetermined system and a compatibility analysis
gives the dimension of $K_d$ depending on certain rank invariants. Those depend on vanishing of some relative 
invariants. Since the construction involves only invariant algebraic operations and all absolute polynomial invariants vanish, 
there are only few possibilities and the answer for higher $d$ might be the same as that for $d=1$. 
This is indeed confirmed by what we have investigated. 

The nonexistence of polynomial integrals of low degree raises the question whether the geodesic flow of 
metrics \eqref{KMm} is integrable. Depending on the class of admissible integrals the methods to approach
this problem are: differential Galois theory, Painlev\'{e} test, numerical simulations. None of these 
have been done yet.

\bigskip

{\bf Acknowledgment.} 
The authors thank Simon King from the University of Jena for the crucial suggestion to use the LinBox package. 
BK thanks Vladimir Matveev for useful discussion and collaboration within RCN-DAAD project 2020-2021
``Differential-Geometric Structures: Invariants and Integrals''.
WS thanks his fellow student Alessandro Schena for help regarding the implementation of the 
`Exploiting Sparsity' point in Maple.

The research leading to our results has received funding from 
the Norwegian Financial Mechanism 2014-2021 (project registration number 2019/34/H/ST1/00636), 
the Polish National Science Centre (NCN grant number 2018/29/B/ST1/02583), 
and the Tromsø Research Foundation (project “Pure Mathematics in Norway”).
This work is an extension and elaboration of \cite{S}.

\end{document}